\documentclass[11pt,reqno]{amsart}%
\usepackage{graphicx,adjustbox}
\usepackage{CJK}
\usepackage{stmaryrd}
\usepackage{mathtools, amsfonts, amssymb, amsthm}
\usepackage{enumerate, url}
\usepackage[margin=1in]{geometry}
\usepackage{caption}
\usepackage{algorithm}
\usepackage{algpseudocode}
\usepackage{mleftright}
\usepackage[normalem]{ulem}
\usepackage{hyperref}
\usepackage{tabu}
\usepackage{tikz-cd}
\usepackage{tikz}
\usetikzlibrary{calc,positioning}
\usepackage{xcolor}
\usetikzlibrary{shapes,arrows}
\usetikzlibrary{decorations.pathmorphing, decorations.text}
\usepackage{empheq}

\usepackage{xcolor}
\hypersetup{
    colorlinks,
    linkcolor={red!50!black},
    citecolor={blue!50!black},
    urlcolor={blue!80!black}
}

\let\latexcirc=\circ
\newcommand{\ccirc}{\mathbin{\mathchoice
  {\xcirc\scriptstyle}
  {\xcirc\scriptstyle}
  {\xcirc\scriptscriptstyle}
  {\xcirc\scriptscriptstyle}
}}
\newcommand{\xcirc}[1]{\vcenter{\hbox{$#1\latexcirc$}}}
\let\circ\ccirc

\makeatletter
\def\Ddots{\mathinner{\mkern1mu\raise\p@
\vbox{\kern7\p@\hbox{.}}\mkern2mu
\raise4\p@\hbox{.}\mkern2mu\raise7\p@\hbox{.}\mkern1mu}}
\makeatother

\newtheorem{theorem}{Theorem}[section]
\newtheorem{proposition}[theorem]{Proposition}
\newtheorem{proposition/definition}[theorem]{Proposition/Definition}
\newtheorem{lemma}[theorem]{Lemma}
\newtheorem{corollary}[theorem]{Corollary}

\newtheorem{prop}[theorem]{Proposition}
\theoremstyle{definition}

\newtheorem{defn}[theorem]{Definition}

\theoremstyle{remark}

\newcommand{\tp}{{\scriptscriptstyle\mathsf{T}}}

\let\O\undefined

\DeclareMathOperator{\O}{O}

\DeclareMathOperator{\rank}{rank}

\DeclareMathOperator{\spn}{span}

\DeclareMathOperator{\Aut}{Aut}
\DeclareMathOperator{\Gal}{Gal}

\newcommand{\Rmnum}[1]{\expandafter\@slowromancap\Romannumeral #1@}

\tikzset{
  symbol/.style={
    draw=none,
    every to/.append style={
      edge node={node [sloped, allow upside down, auto=false]{$#1$}}}
  }
}
\tikzstyle{startstop} = [rectangle, rounded corners, minimum width = 2cm, minimum height=1cm,text centered, draw = black, fill = red!40]
\tikzstyle{process} = [rectangle, minimum width=1cm, minimum height=1cm, text centered, draw=black, fill = yellow!50]
\tikzstyle{processb} = [rectangle, minimum width=1cm, minimum height=1cm, text centered, draw=black, fill = blue!50]
\tikzstyle{arrow} = [->,>=stealth]
\tikzstyle{dotsnd}=[rectangle, minimum width=1cm, minimum height=1cm, text centered]

\tikzstyle{processd} = [dashed, minimum width=8cm, minimum height=1.5cm, text centered, draw=black]
\tikzstyle{processdd} = [dashed, minimum width=7cm, minimum height=1.5cm, text centered, draw=black]

\begin{document}
\title{skew-sparse matrix multiplication}
\author[Huang Ye and Gao]{Qiao-Long Huang$^1$, Ke Ye$^2$, Xiao-Shan Gao$^2$\\$^1$ Research Center for Mathematics and Interdisciplinary Sciences, Shandong University\\$^2$ KLMM, UCAS, Academy of Mathematics and Systems Science,
Chinese Academy of Sciences}

\begin{abstract}
Based on the observation that $\mathbb{Q}^{(p-1) \times (p-1)}$ is isomorphic to a quotient skew polynomial ring, we propose a new method for $(p-1)\times (p-1)$ matrix multiplication over $\mathbb{Q}$, where $p$ is a prime number. The main feature of our method is the acceleration for matrix multiplication if the product is skew-sparse. Based on the new method, we design a deterministic algorithm with complexity $\O(T^{\omega-2} p^2)$, where $T\le p-1$ is a parameter determined by the skew-sparsity of input matrices and $\omega$ is the asymptotic exponent of matrix multiplication. Moreover, by introducing randomness, we also propose a probabilistic algorithm with complexity $\O^\thicksim(t^{\omega-2}p^2+p^2\log\frac{1}{\nu})$, where $t\le p-1$ is the skew-sparsity of the product and $\nu$ is the probability parameter.
\end{abstract}

%\subjclass[]{}

%\keywords{}

\maketitle

\section{introduction}
Fast algorithms for matrix multiplication are essential ingredients in numerous fundamental computational problems such as linear system solving~\cite{wiedemann1986solving}, matrix inversion~\cite{bunch1974triangular}, determinant evaluation~\cite{keller1985fast} and Boolean matrix multiplication~\cite{fischer1971boolean,Ian1971Efficient}.
Thus the computational complexity of matrix multiplication is one of the most important problems in theoretical computer science and mathematics, which attracts extensive attention from researchers and practitioners in various areas.

The complexity of matrix multiplication is usually measured by \emph{the asymptotic exponent} $\omega$, which is defined by
\[
\omega \coloneqq\liminf_{n\to \infty} \log_n M_n,
\]
where $M_n$ is the total cost of operations for $n\times n$ matrix multiplication. It is clear from the definition that $2 \le \omega \le 3$. For a relatively long time, it was believed that $\omega = 3$. The seminal work \cite{1969Gaussian} of Strassen for the first time proved that $\omega \le 2.81$, after which the upper bound of $\omega$ was further improved by a series of works \cite{HK71, Pan78,Bini80, Strassen86}. Using a powerful technique called the ``laser method" developed in \cite{CoppersmithW90,Williams12,Gall14a}, Alman and Williams were able to obtain the state-of-the-art upper bound $\omega \le 2.37286$ \cite{AlmanW21}. Unfortunately, it was proved that the upper bound of $\omega$ one can obtain by the ``laser method" is at least $2.3078$ \cite{AFL15}, indicating the necessity of developing a new method.

Besides the complexity of matrix multiplication for general matrices, the same problem for structured matrices is also of great importance in both theoretical studies \cite{Gustavson78,Schorr82,YusterZ05,CHILO18,YL18,LY20} and practical applications\cite{CGR88,GR97,OS00,Van00,GV13,DGPRR18}. The most commonly considered structured matrices are circulant, Toeplitz, Hankel, Cauchy, symmetric, sparse matrices and their variants.

In this paper, we discuss fast matrix multiplication for a new type of structured matrices called \emph{skew-sparse} matrices, which are defined by sparse skew polynomials. It is noticeable that the complexity of fast sparse matrix multiplication algorithms \cite{Gustavson78,Schorr82,YusterZ05} can be even higher than $\O(n^{2.37286})$ when applied to dense matrices. Our algorithm for skew-sparse matrix multiplication, as a comparison, has complexity at most $\O(n^{\omega})$.
%%%%%%%%%%%%%%%%%%%%%%%%%%%%%%%%%%%%%%%%%%%%%%%%%%%%
\section{preliminaries}
In this section, we provide some preliminaries on cyclotomic field, normal basis and skew polynomials.
%%%%%%%%%%%%%%%%%%%%%%%%%%%%%%%%%%%%%%%%%%%%%%%%%%%%
\subsection{cyclotomic field}\label{subsec:cyclotomic}
Let $p = n+1$ be a prime and let $\beta$ be a $p$-th primitive root of unity. Hence we have
\[
\beta^{p-1}+\cdots+\beta+1=0.
\]
The extension field $\mathbb{Q}(\beta)$ of $\mathbb{Q}$ is called the \emph{cyclotomic field} generated by $\beta$. It is clear that $[\mathbb{Q}(\beta):\mathbb{Q}] = p - 1$.

A \emph{primitive root of $\mathbb{Z}_p$} is an integer $r\in \{0,1,\dots, p-1\}$ such that $\langle r \rangle = \mathbb{Z}^{\times}_p$. Given a primitive root of $\mathbb{Z}_p$, we consider a $\mathbb{Q}$-homomorphism
\begin{equation}\label{eq:sigma}
\sigma: \mathbb{Q}(\beta) \to \mathbb{Q}(\beta),\quad \beta\mapsto \beta^r.
\end{equation}
Since $r$ is a primitive root of $\mathbb{Z}_{p}$, we have
\[
\{1,2,\dots,p-1\} \equiv \{1,r,\dots,r^{p-2}\}\pmod{p}.
\]
This implies that $\sigma$ permutes $\beta,\beta^2,\dots,\beta^{p-1}$, which are roots of the irreducible polynomial $x^{p-1}+x^{p-2}+\cdots+x+1$. Thus $\sigma$ is an automorphism of $\mathbb{Q}(\beta)$ over $\mathbb{Q}$, i.e., $\sigma\in \Aut(\mathbb{Q}(\beta)/\mathbb{Q})$.
%%%%%%%%%%%%%%%%%%%%%%%%%%%%%%%%%%%%%%%%%%%%%%%%%%%%
\subsection{normal basis}\label{subsec:normal basis}
We recall that given a Galois extension $F\subseteq K$ with the Galois group $G$. A \emph{normal basis} is a basis of $K/F$, which is of the form $\{g(\beta):g\in G\}$ for some $\beta\in K$. We notice that the Galois group of $\mathbb{Q}(\beta)/\mathbb{Q}$ is $\Gal(\mathbb{Q}(\beta)/\mathbb{Q}) =\langle\sigma\rangle$, where $\sigma$ is the automorphism defined in \eqref{eq:sigma}. For each $i = 1,\dots, p-1$, we denote
\begin{equation}\label{eq:normal basis}
v_i \coloneqq \sigma^i(\beta) =  \beta^{r^{i-1}}.
\end{equation}
It is straightforward to verify that $v_{i+1} = \sigma(v_i)$. According to the discussion in Subsection~\ref{subsec:cyclotomic} and the construction of $v_i$'s, we have the following:
\begin{lemma}\label{lem:normal basis}\cite[page~4]{gao1993normal}
$\{v_1,v_2,\dots,v_{p-1}\}$ is a normal basis of $\mathbb{Q}(\beta)/\mathbb{Q}$.
\end{lemma}
Since each $v_i$ is a  power of $\beta$, with respect to the basis $\{v_1,v_2,\dots,v_{p-1}\}$, there exist fast algorithms \cite{P2010Algebraic} to compute multiplication,
division, addition and subtraction in $\mathbb{Q}(\beta)$. Moreover, by identifying $\mathbb{Q}(\beta)$ with $\mathbb{Q}[z]/(H(z))$, where
$H(z)=z^{p-1}+\cdots+1$, we have the following observation regarding the complexity of arithmetic operations in $\mathbb{Q}(\beta)$:
\begin{lemma}\cite{1991On}\label{lem:arithmetic operation}
One arithmetic operation  in $\mathbb{Q}(\beta)$ equals to $\O^\thicksim(p)$ arithmetic operations in $\mathbb{Q}$.
\end{lemma}
%%%%%%%%%%%%%%%%%%%%%%%%%%%%%%%%%%%%%%%%%%%%%%%%%%%%
\subsection{skew polynomials}
Let $x$ be an indeterminate and let $\sigma$ be the automorphism defined in \eqref{eq:sigma}. We denote by $\mathbb{Q}(\beta)[x;\sigma]$ the ring $(\mathbb{Q}(\beta)[x],+,\ast_\sigma)$, where $+$ is the usual polynomial addition and $\ast_\sigma$ is induced by $x \ast_{\sigma} c \coloneqq \sigma(c)x$ and $c \ast_{\sigma} x \coloneqq cx$. More precisely, we have
\begin{align*}
\sum_{i=0}^d a_i x^i + \sum_{i=0}^d b_i x^i &= \sum_{i=0}^d (a_i + b_i) x^i, \\
\left( \sum_{j=0}^d a_j x^j \right) \ast_{\sigma} \left(  \sum_{k=0}^e b_k x^k\right) &=\sum_{\ell=0}^{d+e} \sum_{s=0}^{\ell} a_s\sigma^{s}(b_{\ell-s}) x^\ell. \\
\end{align*}
Here in the second formula, we adopt the convention that $a_i =0$ (resp. $b_i=0$) if $i>d$ (resp. $i>e$). $\mathbb{Q}(\beta)[x;\sigma]$ is called \emph{the ring of skew polynomials}. According to \cite[Lemma~1.4]{CL17}, there exists an isomorphism of algebras:
\begin{align}\label{eq:isomorphism}
\varepsilon: \mathbb{Q}(\beta)[x;\sigma]/(x^{p-1}-1) &\to \mathrm{End}_{\mathbb{Q}}(\mathbb{Q}(\beta)), \\
\sum_{i=0}^{p-2} a_i x^i &\mapsto \sum_{i=0}^d a_i \sigma^i. \nonumber
\end{align}
Here an element in $ \mathbb{Q}(\beta)[x;\sigma]/(x^{p-1}-1)$ is written as a polynomial of degree at most $p-2$ and we adopt this convention in the rest of the paper.
%%%%%%%%%%%%%%%%%%%%%%%%%%%%%%%%%%%%%%%%%%%%%%%%%%%%
\subsubsection{sumset of skew polynomials}
In the sequel, we need the notion of the sparsity of a skew polynomial. Given a polynomial $f =\sum_{i=1}^t a_i x^{e_i}\in \mathbb{Q}(\beta)[x;\sigma]/(x^{p-1}-1)$ where $0 \le e_1 < \cdots < e_t \le p-2$ and all $a_i \ne 0, i=1,\dots, t$, the set $\mathrm{supp}(f) \coloneqq \{e_1,\dots,e_t\}$ is called the \emph{support} of $f$. Consequently, we define the \emph{sparsity} of $f$ to be $\# f \coloneqq \# \mathrm{supp}(f) = t$.

Given two polynomials $f$ and $g$ in $\mathbb{Q}(\beta)[x;\sigma]/(x^{p-1}-1)$, we define the \emph{sumset of $f$ and $g$} by
\begin{equation}\label{eq:sumset}
\mathbb{S}(f,g) \coloneqq \{e_{f}+e_{g}:  e_{f}\in\mathrm{supp}(f),e_{g}\in\mathrm{supp}(g) \} \pmod{p-1}.
\end{equation}

\begin{lemma}\label{lem:sparsity}
We have the following:
\begin{enumerate}[(i)]
\item\label{lem:sparsity:item1} $\mathbb{S}(f,g)$ contains $\mathrm{supp}(f\ast_\sigma g)$;
\item\label{lem:sparsity:item2} $\# (f\ast_\sigma g)\leq \#\mathbb{S}(f,g) \le
\#f \cdot \#g$ and $\#\mathbb{S}(f,g) \le p-1$;
\item\label{lem:sparsity:item3} the strict inequality $\#(f\ast_\sigma g) <\#\mathbb{S}(f,g)$ holds only in the occurrence of coefficient cancellations.
\end{enumerate}
\end{lemma}

%%%%%%%%%%%%%%%%%%%%%%%%%%%%%%%%%%%%%%%%%%%%%%%%%%%%
\subsubsection{evaluation of a skew polynomial}\label{sub-sec-1}
According to \eqref{eq:isomorphism}, a skew polynomial $f\in \mathbb{Q}(\beta)[x;\sigma]/(x^{p-1}-1)$ defines a $\mathbb{Q}$-linear mapping $\epsilon(f): \mathbb{Q}(\beta)\rightarrow \mathbb{Q}(\beta)$ obtained by evaluating $f$ at $\sigma$. For simplicity, we denote $\epsilon(f)(b)$ by $f(b)$ for each $b \in \mathbb{Q}(\beta)$.
%since $\sigma$ has order $n$, $f(a)$ is actually well-defined for $f$ in $\mathbb{Q}(\beta)[x;\sigma]/(x^n-1)$.

Next we briefly discuss a fast approach to simultaneously evaluate $f(b)$ for $b \in \mathbb{Q}(\beta)$. Since $\mathbb{Q}(\beta)$ is a $p-1$ dimensional vector space over $\mathbb{Q}$, we may write
\[
b = \sum_{j=1}^{p-1} b_{j}v_j,
\]
where $b_{j} \in \mathbb{Q}$ and $\{v_1,\dots, v_{p-1}\}$ is the normal basis defined in \eqref{eq:normal basis}. In particular, we have $f(v_i)=\sum_{j=1}^{p-1} a_{ij}v_j$ for some $a_{ij}\in \mathbb{Q}$, $1\le i,j \le p - 1$. Thus with respect to the normal basis $\{v_j\}_{j=1}^{p-1}$, the $\mathbb{Q}$-linear map $\epsilon(f)$ can be represented by the $(p-1)\times (p-1)$ matrix over $\mathbb{Q}$:
\[
\epsilon(f) = \begin{bmatrix}
a_{11}&a_{12}&\cdots&a_{1,p-1}\\
a_{21}&a_{22}&\cdots&a_{2,p-1}\\
\vdots&\vdots&\ddots&\vdots\\
a_{p-1,1}&a_{p-1,2}&\cdots&a_{p-1,p-1}\\
\end{bmatrix}.
\]
The above observation enables us to write \eqref{eq:isomorphism} in a more explicit way. Namely, we have
\begin{align}\label{eq:poly-matrix}
  \varphi: \mathbb{Q}(\beta)[x;\sigma]/(x^{p-1}-1)  & \to \mathbb{Q}^{(p-1) \times (p-1)}        \\
   f &\mapsto (a_{ij})_{i,j=1}^{p-1} \nonumber
\end{align}
where $f(v_i)=\sum_{j=1}^{p-1} a_{ij}v_j, 1\le i \le p-1$.

Accordingly, the evaluation of $f$ at $b = \sum_{j=1}^{p-1} b_j v_j \in \mathbb{Q}(\beta)$ can be written as
\[
f(b)= \begin{bmatrix}
b_1& \cdots &b_{p-1}
\end{bmatrix}
\begin{bmatrix}
f(v_1)\\
\vdots\\
f(v_{p-1})\\
\end{bmatrix} = \begin{bmatrix}
b_1&\cdots&b_{p-1}
\end{bmatrix}
\begin{bmatrix}
a_{1,1}&\cdots&a_{1,p-1}\\
\vdots&\ddots&\vdots\\
a_{p-1,1}&\cdots&a_{p-1,p-1}\\
\end{bmatrix}
\begin{bmatrix}
v_1\\
\vdots\\
v_{p-1}\\
\end{bmatrix}.
\]
It is straightforward to verify that for $b_1,\dots, b_t\in \mathbb{Q}(\beta)$, we have
\begin{equation}\label{eq:evaluation}
\left[
\begin{array}{c}
f(b_1) \\
f(b_2) \\
\vdots \\
f(b_t)\\
\end{array}
\right]=
\begin{bmatrix}
b_{1,1}&\cdots&b_{1,p-1}\\
\vdots&\ddots&\vdots\\
b_{t,1}&\cdots&b_{t,p-1}\\
\end{bmatrix}
\begin{bmatrix}
a_{1,1}&\cdots&a_{1,p-1}\\
\vdots&\ddots&\vdots\\
a_{p-1,1}&\cdots&a_{p-1,p-1}\\
\end{bmatrix}
\begin{bmatrix}
v_1 \\
\vdots \\
v_{p-1}\\
\end{bmatrix},
\end{equation}
where  $b_i = \sum_{j=1}^{p-1} b_{ij}v_j$ and $1\le i \le t$. We notice that in \eqref{eq:evaluation}, the two matrices have size of $t\times (p-1)$ and $(p-1)\times (p-1)$ respectively. Thus we have the following:
\begin{lemma}\label{lema:eval}\cite{giesbrecht2020sparse}
The evaluation of $f(b_1),\dots, f(b_t)$ can be done by $\O(t^{\omega-2}p^2)$  operations in $\mathbb{Q}$, where $\omega$ is the exponent of the matrix multiplication.
\end{lemma}
%%%%%%%%%%%%%%%%%%%%%%%%%%%%%%%%%%%%%%%%%%%%%%%%%%%%
\subsubsection{interpolation of a skew polynomial}
Assume $f\in\mathbb{Q}(\beta)[x;\sigma]$ and $\mathrm{supp}(f)=\{e_1,\dots,e_t\}$. We discuss how to find coefficients of $f$ from $f(1), f(v_1),\dots, f(v_1^{t-1})$.
%For the convenience of description, denote $w_{i}:=v_{i+1},i=0,1,\dots,n-1$. Now we present how to find the coefficients of $f$ from the evaluations $f(w^i_0),i=0,\dots,t-1$.

By assumption, $f$ has the form $f = \sum_{i=1}^t c_i x^{e_i}$. Since
\[
\sigma^k(v^i_1) = (\sigma^k(v_1))^i=(v_{k+1})^i,\quad k \in \mathbb{N},
\]
we have
\[
f(v^i_1)=c_1(v_{e_1+1})^i+\cdots+c_t(v_{e_t+1})^i,\quad i=0,1,\dots,t-1.
\]
Here indexes $e_1+1,\dots, e_t+1$ are are taken modulo $p$. This implies
\begin{equation}\label{eq:interpolation}
\left[
\begin{array}{c}
f(v^0_1) \\
f(v^1_1) \\
\vdots \\
f(v^{t-1}_1)\\
\end{array}
\right]=\left[\begin{array}{cccc}
1&1&\cdots&1\\
v_{e_1+1}&v_{e_2+1}&\cdots&v_{e_t+1}\\
\vdots&\vdots&\ddots&\vdots\\
v^{t-1}_{e_1+1}&v^{t-1}_{e_2+1}&\cdots&v^{t-1}_{e_t+1}\\
\end{array}\right]
\left[
\begin{array}{c}
c_1 \\
c_2 \\
\vdots \\
c_t\\
\end{array}
\right].
\end{equation}
To recover $c_1,\dots,c_t$, we need to solve the linear system \eqref{eq:interpolation}, whose coefficient matrix is Vandermonde. Therefore we obtain the result that follows.
\begin{lemma}\cite{giesbrecht2020sparse}
The interpolation of $f$ can be done by $\O^\sim(tn)$ operations in $\mathbb{Q}$.
\end{lemma}
%%%%%%%%%%%%%%%%%%%%%%%%%%%%%%%%%%%%%%%%%%%%%%%%%%%%
\section{deterministic matrix multiplication via skew polynomials}
\subsection{relation between matrices and skew polynomials}\label{subsec:relation}
Let $\varphi: \mathbb{Q}(\beta)[x;\sigma]/(x^{p-1}-1) \to \mathbb{Q}^{(p-1)\times (p-1)}$ be the $\mathbb{Q}$-algebra isomorphism defined in \eqref{eq:poly-matrix}. In this subsection, we discuss algorithms computing $\varphi$ and $\varphi^{-1}$. According to the definition of $\varphi$ and properties of normal basis, it is easy to establish the following relation between skew polynomials and matrices.
\begin{lemma}\label{lem:relation}
Let $f = \sum_{j=0}^{p - 2} \mu_{j+1} x^j\in \mathbb{Q}(\beta)[x;\sigma]/(x^{p-1}-1)$ be a skew polynomial and let $C = (c_{ij})_{i,j=1}^{p-1}$ be a $(p-1)\times (p-1)$ matrix over $\mathbb{Q}$. Then $\varphi(f) = C$ if and only if
\begin{equation}\label{eq-2}
\begin{bmatrix}
v_1 & v_2 & \cdots & v_{p-1} \\
v_2 & v_3 & \cdots & v_1\\
\vdots & \vdots & \ddots & \vdots\\
v_{p-1} & v_1 & \cdots & v_{p -2}\\
\end{bmatrix}
\begin{bmatrix}
\mu_1 \\
\mu_2 \\
\vdots \\
\mu_{p-1}\\
\end{bmatrix}
= \begin{bmatrix}
c_{11} & c_{12} & \cdots & c_{1,p-1}\\
c_{21} & c_{22} & \cdots & c_{2,p-1}\\
\vdots & \vdots & \ddots & \vdots\\
c_{p-1,1} & c_{p-1,2} & \cdots & c_{p-1,p-1}\\
\end{bmatrix}
\begin{bmatrix}
v_1 \\
v_2 \\
\vdots \\
v_{p-1}\\
\end{bmatrix},
\end{equation}
where $\{v_1,\dots, v_{p-1}\}$ is the normal basis defined in \eqref{eq:normal basis}.
\end{lemma}

\subsubsection{from matrices to skew polynomials}
Given a $(p-1)\times (p-1)$ matrix $C = (c_{ij})_{i,j=1}^{p-1}$ over $\mathbb{Q}$, we may compute the skew polynomial $f \coloneqq \varphi^{-1}(C)$.

%\begin{algorithm}[!h]
%\caption{converting matrices to skew polynomials}
%\label{alg-1}
%\begin{algorithmic}[1]
%\renewcommand{\algorithmicrequire}{\textbf{Input}:}\Require $(p-1) \times (p-1)$ matrix $C$ and normal basis $\{v_1,\dots,v_n\}$ of $\mathbb{Q}(\beta)/\mathbb{Q}$.
%\renewcommand{\algorithmicensure}{\textbf{Output}:}\Ensure the skew polynomial $\varphi^{-1}(C)$
%
%\State compute $\beta_{i} \coloneqq \sum_{j=1}^{p-1} c_{ij} v_j, i=1,\dots,p-1$.
%
%\State compute coefficients $\mu_1,\dots,\mu_{p-1}$ by solving the linear system.
%\begin{equation}\label{eq:alg-1:circulant}
%\begin{bmatrix}
%v_1&v_2&\cdots&v_{p-1}\\
%v_2&v_3&\cdots&v_1\\
%\vdots&\vdots&\ddots&\vdots\\
%v_{p-1}&v_1&\cdots&v_{p - 2}
%\end{bmatrix}
%\begin{bmatrix}
%\mu_1 \\
%\mu_2 \\
%\vdots \\
%\mu_{p-1}
%\end{bmatrix} =
%\begin{bmatrix}
%\beta_1 \\
%\beta_2 \\
%\vdots \\
%\beta_{p-1}\\
%\end{bmatrix}.
%\end{equation}
%
%\State return $f= \sum_{j=0}^{p-2} \mu_{j+1} x^j$.
%\end{algorithmic}
%\end{algorithm}

\begin{lemma}\label{lm-1}
Let $\{v_i\}_{i=1}^{p-1}$ be as above. We denote
\[
V \coloneqq \begin{bmatrix}
v_1 & v_2 & \cdots & v_{p-1}\\
v_2 & v_3 & \cdots & v_1\\
\vdots & \vdots & \ddots & \vdots\\
v_{p-1} & v_1 & \cdots & v_{p-2}
\end{bmatrix}, \quad
W \coloneqq \begin{bmatrix}
\frac{1}{v_1}-1 & \frac{1}{v_2}-1 & \cdots & \frac{1}{v_{p-1}}-1\\
\frac{1}{v_2}-1 & \frac{1}{v_3}-1 & \cdots & \frac{1}{v_1}-1\\
\vdots & \vdots & \ddots & \vdots\\
\frac{1}{v_{p-1}}-1 & \frac{1}{v_1}-1 & \cdots & \frac{1}{v_{p - 2}}-1.
\end{bmatrix}
\]
Then we have $VW=p I_{p-1}$.
\end{lemma}

\begin{proof}
We proceed by directly computing the $(i,j)$-th entry $a_{ij}$ of $VW$. We want to prove that
\[
a_{ij}= p \delta_{ij},
\]
where $\delta_{ij}$ is the Kronecker delta. We observe that the $i$-th row of $V$ is $(v_i,v_{i+1},\dots,v_{i+n-1})$ and $j$-th column of $W$ is $(\frac{1}{v_j}-1,\frac{1}{v_{j+1}}-1,\dots,\frac{1}{v_{j+p-2}}-1)^T$. Thus we have
\[
a_{ij}=\sum_{k=0}^{p-2}v_{i+k} \left( \frac{1}{v_{j+k}}-1 \right)=\sum_{k=0}^{p-2}\frac{v_{i+k}}{v_{j+k}}-\sum_{k=0}^{p-2}
v_{i+k}  =\sum_{k=0}^{p - 2}\sigma^k\left( \frac{v_i}{v_j} \right)-\sum_{k=0}^{p - 2}v_k.
\]
Here indexes are taken modulo $p$.
Since $\sum_{j=0}^{p-1} \beta^{j} = 0$, we conclude that for each $i\in \mathbb{Z}$,
\[
\sum_{k=0}^{p-2}v_{i+k}=\sum_{k=0}^{p-2}v_{k} = -1.
\]
%
%According to the chosen of $\beta$,  $\sum_{k=0}^{n-1}v_k=\sum_{k=0}^{n-1}\beta^k=-1$. So we have $$\sum_{k=0}^{n-1}v_{i+k}=-1$$

If $i=j$, then $a_{ij}= \left( \sum_{k=0}^{p - 2}\sigma^k(1) \right)+1= p$.
If $i\neq j$, then $\frac{v_i}{v_j}=\beta^{r^{i-1}-r^{j-1}}\neq 0$. As $r$ is a primitive root of $\mathbb{Z}_p$, there exists an $m\in\{1,2,\dots,p-1\}$ such that  $r^{m-1}=r^{i-1}-r^{j-1} \pmod{p}$. Hence $\frac{v_i}{v_j}=v_{m}$ and
\[
a_{ij}=\left(  \sum_{k=0}^{p - 2}\sigma^k(v_m) \right) +1=\left(  \sum_{k=0}^{p - 2}v_{m+k} \right) +1 =0.
\]
\end{proof}

By a combination of Lemmas \ref{lem:relation} and \ref{lm-1}, we derive the following formula for $\varphi^{-1}(C)$.
\begin{lemma}
Let $C = (c_{ij})_{i,j=1}^{p-1}$ be a $(p-1)\times (p-1)$ matrix over $\mathbb{Q}$. Then $\varphi^{-1}(C) = \sum_{j=0}^{p - 2} \mu_{j+1} x^j$ where
\begin{equation}\label{eq-1}
\begin{bmatrix}
\mu_1 \\
\mu_2 \\
\vdots \\
\mu_{p-1}
\end{bmatrix}
=\frac{1}{p}
\begin{bmatrix}
\frac{1}{v_1}-1 & \frac{1}{v_2}-1 & \cdots & \frac{1}{v_{p-1}}-1\\
\frac{1}{v_2}-1 & \frac{1}{v_3}-1 & \cdots & \frac{1}{v_1}-1\\
\vdots & \vdots & \ddots & \vdots\\
\frac{1}{v_{p-1}}-1 & \frac{1}{v_1}-1 & \cdots & \frac{1}{v_{p-2}}-1
\end{bmatrix}
\begin{bmatrix}
c_{11} & c_{12} & \cdots & c_{1,p-1}\\
c_{21} & c_{22} & \cdots & c_{2,p-1}\\
\vdots & \vdots & \ddots & \vdots\\
c_{p-1,1} & c_{p-1,2} & \cdots & c_{p-1,p-1}
\end{bmatrix}
\begin{bmatrix}
v_1 \\
v_2 \\
\vdots \\
v_{p-1}
\end{bmatrix}.
\end{equation}
\end{lemma}

\begin{prop}\label{pro-1}
The skew polynomial $\varphi^{-1}(C)$ can be computed by $\O^\sim(p^2)$ operations in $\mathbb{Q}$.
\end{prop}
\begin{proof}
Since \eqref{eq-1} only involves the $(p-1)\times (p-1)$ matrix vector product, the complexity is at most $\O^\sim(p^2)$.
%Since the coefficient matrix in \eqref{eq:alg-1:circulant} is a circulant matrix, we can solve the equation in time $O^\sim(n^2)$. \red{[Qiaolong: Please add a reference here.]}
\end{proof}

\subsubsection{from skew polynomials to matrices}
Assume $f\in\mathbb{Q}(\beta)[x;\sigma]/(x^{p-1}-1)$ has the form
\begin{equation}\label{eq-4}
f:=\mu_1+\mu_2x+\mu_3x^2+\cdots+\mu_{p-1}x^{p-2}.
\end{equation}

Now we present the procedure $\varphi$ that transforms the skew polynomial into the matrix in pseudocode.

\begin{algorithm}[!ht]
\caption{matrices to skew polynomials}
\label{alg-2}
\begin{algorithmic}[1]
\renewcommand{\algorithmicrequire}{\textbf{Input}:}\Require skew polynomial $f = \sum_{j=0}^{p-2} \mu_{j+1} x^j$ and normal basis $\{v_1,\dots,v_{p-1}\}$ of $\mathbb{Q}(\beta)/\mathbb{Q}$.

\renewcommand{\algorithmicensure}{\textbf{Output}:}\Ensure the matrix $\varphi(f) = (c_{ij})_{i,j=1}^{p-1}$.

\State \label{alg-2:step1} compute
\[
\begin{bmatrix}
\beta_1 \\
\beta_2 \\
\vdots \\
\beta_{p-1}\\
\end{bmatrix} =
\begin{bmatrix}
v_1&v_2&\cdots&v_{p-1}\\
v_2&v_3&\cdots&v_1\\
\vdots&\vdots&\ddots&\vdots\\
v_{p-1}&v_1&\cdots&v_{p-2}\\
\end{bmatrix}
\begin{bmatrix}
\mu_1 \\
\mu_2 \\
\vdots \\
\mu_{p-1}\\
\end{bmatrix}.
\]
\State \label{alg-2:step2} compute $c_{ij}\in \mathbb{Q},1\le i,j \le p-1$ such that
\[
\begin{bmatrix}
\beta_1 \\
\beta_2 \\
\vdots \\
\beta_{p-1}\\
\end{bmatrix}
= \begin{bmatrix}
c_{11} & c_{12} & \cdots & c_{1,p-1}\\
c_{21} & c_{22} & \cdots & c_{2,p-1}\\
\vdots & \vdots & \ddots & \vdots\\
c_{p-1,1} & c_{p-1,2} & \cdots & c_{p-1,p-1}\\
\end{bmatrix}
\begin{bmatrix}
v_1 \\
v_2 \\
\vdots \\
v_{p-1}\\
\end{bmatrix},
\]
\end{algorithmic}
\end{algorithm}

\begin{prop}\label{pro-2}
Algorithm \ref{alg-2} returns $\varphi(f)$ using  $\O^\sim(p^2)$ operations in
$\mathbb{Q}$.
\end{prop}
\begin{proof}
Since the matrix in Step~\ref{alg-2:step1} is a circulant matrix, we compute the product in time $\O^\sim(p)$ over $\mathbb{Q}(\beta)$, which is $\O^\sim(p^2)$ operations in $\mathbb{Q}$.
%[Qiaolong: please check this proof.]}
\end{proof}
%%%%%%%%%%%%%%%%%%%%%%%%%%%%%%%%%%%%%%%%%%%%%%%%%%%%
\subsection{matrix multiplication via skew polynomials}
Now we are ready to present our algorithm for matrix multiplication. The main idea is to reduce the matrix multiplication to the skew polynomial multiplication.

\begin{algorithm}[!ht]
\caption{matrix multiplication via skew polynomials}
\label{alg-3}
\begin{algorithmic}[1]
\renewcommand{\algorithmicrequire}{\textbf{Input}:}\Require matrices $A,B\in \mathbb{Q}^{(p-1) \times (p-1)}$, where $p$ is a prime.

\renewcommand{\algorithmicensure}{\textbf{Output}:}\Ensure the product $A B$.

\State find a $p$-th primitive root of unity $\beta$.

\State find a primitive root $r$ of $\mathbb{Z}_{p}$ and compute $v_i=\beta^{r^{i-1}}, i=1,\dots,p-1$. %and $w_{i-1}:=v_i,$.

\State compute $\varphi^{-1}(A),\varphi^{-1}(B)$ by \eqref{eq-1} w.r.t the normal basis $\{v_1,\dots,v_{p-1}\}$. \label{alg-3:step3}

\State compute the sumset $\mathbb{S}(\varphi^{-1}(A),\varphi^{-1}(B)) = \{e_1,e_2,\dots,e_t\}$. \label{alg-3:step4}

\State compute $\mathrm{Eval}_i:=\left( \varphi^{-1}(A)\ast_{\sigma} \varphi^{-1}(B)\right) (v^i_1),i=0,1,\dots,t-1$. \label{alg-3:step5}

\State solve
\begin{equation*}
\begin{bmatrix}
\mathrm{Eval}_0 \\
\mathrm{Eval}_1 \\
\vdots \\
\mathrm{Eval}_{t-1}\\
\end{bmatrix}
=
\begin{bmatrix}
1&1&\cdots&1\\
v_{e_1+1}&v_{e_2+1}&\cdots&v_{e_t+1}\\
\vdots&\vdots&\ddots&\vdots\\
v^{t-1}_{e_1+1}&v^{t-1}_{e_2+1}&\cdots&v^{t-1}_{e_t+1}
\end{bmatrix}
\begin{bmatrix}
c_1 \\
c_2 \\
\vdots \\
c_t
\end{bmatrix}
\end{equation*}
and set $f \coloneqq c_1x^{e_1}+\cdots+c_tx^{e_t}$. \label{alg-3:step6}

\State compute $\varphi(f)$ by Algorithm~\ref{alg-2}. \label{alg-3:step7}
\end{algorithmic}
\end{algorithm}

\begin{proposition}\label{prop:matrix multiplication}
Algorithm~\ref{alg-3} computes $AB$ using  $\O^\sim(t^{\omega-2}p^2)$ operations in
$\mathbb{Q}$, where $t$ is the cardinality of the sumset of $\varphi^{-1}(A)$ and $\varphi^{-1}(B)$.
\end{proposition}
\begin{proof}
If we denote by $m$ (resp. $\ast_\sigma$) the matrix (resp. skew-polynomial) multiplication, then we have the following commutative diagram, which validates Algorithm~\ref{alg-3}.
\begin{figure}[!ht]
  \begin{tikzcd}[row sep= 10ex, column sep=large]
     \mathbb{Q}^{(p-1)\times (p-1)} \times \mathbb{Q}^{(p-1)\times (p-1)}   \arrow{d}{m} \arrow{r}{\varphi^{-1}\times \varphi^{-1}} & \mathbb{Q}(\beta)[x;\sigma]/(x^{p-1}-1) \times \mathbb{Q}(\beta)[x;\sigma]/(x^{p-1}-1) \arrow{d}{\ast_{\sigma}}   \\
     \mathbb{Q}^{(p-1)\times (p-1)}   & \mathbb{Q}(\beta)[x;\sigma]/(x^{p-1}-1) \arrow{l}{\varphi}
  \end{tikzcd}
\end{figure}

Next we analyse the complexity of Algorithm~\ref{alg-3}.
If $r$ is a primitive root of $p$, then $r^1,r^2,\dots, r^{p-1}\pmod{p}$ must generate distinct integers from $1$ to $(p-1)$. Checking the membership of a number in $\{2,\dots,p-2\}$ has complexity at most $\O(p^2)$ over $\mathbb{Q}$.
  By Proposition \ref{pro-1}, the complexity of Step~\ref{alg-3:step3} in Algorithm~\ref{alg-3} is $\O^\sim(p^2)$  over $\mathbb{Q}$.
  It is straightforward to verify that the complexity of Step~\ref{alg-3:step4} is $\O(p^2)$ over $\mathbb{Q}$.
  The cost of Step~\ref{alg-3:step5} is $\O^\sim(t^{\omega-2}p^2)$ $\mathbb{Q}$-operations by Lemma~\ref{lema:eval}.
  In Step~\ref{alg-3:step6}, since the coefficient matrix is a Vandermonde matrix, solving this linear system costs $\O^\sim(tp)$ $\mathbb{Q}$-operations~\cite{kaltofen1988improved}.
 By Proposition~\ref{pro-2}, Step~\ref{alg-3:step7} costs $\O^\thicksim(p^2)$ operations in $\mathbb{Q}$.
\end{proof}

%%%%%%%%%%%%%%%%%%%%%%%%%%%%%%%%%%%%%%%%%%%%%%%%%%%%
\subsection{skew-sparse matrix multiplication}
Suppose $p$ is a prime number and $r$ is a primitive root of $\mathbb{Z}_{p}$. Let $1 \le k \le p-1$ be the integer such that $r^{k-1}\equiv p-1 \pmod{p}$. For each $1 \le j\ne k \le p-1$, we denote by $s_j$ the integer between $1$ and $(p-1)$ such that  $r^{s_j-1} \equiv r^{j-1} + 1 \pmod{p}$. By Lemma~\ref{lem:normal basis}, $\{v_j\}_{j =1}^{p-1}$ is a normal basis of $\mathbb{Q}(\beta)/\mathbb{Q}$ where $v_i = \beta^{r^{j -1}}, 1 \le j \le p-1$.

We construct two $(p-1) \times (p-1)$ matrices
\[
X \coloneqq \begin{bmatrix}
0 & 1 & 0  & \cdots & 0 \\
0 & 0 & 1  & \cdots & 0 \\
\vdots & \vdots & \vdots & \ddots & \vdots \\
0 & 0 & 0 & \cdots & 1 \\
1 & 0 & 0 & \cdots & 0
\end{bmatrix},\quad Y \coloneqq \begin{bmatrix}
E_{s_1} \\
\vdots \\
E_{s_{k-1}} \\
J \\
E_{s_{k+1}} \\
\vdots \\
E_{s_{p-1}} \\
\end{bmatrix},
\]
where $E_{i}$ is the row vector whose elements are all $0$ except the $i$-th, which is equal to $1$, and $J$ is the row vector whose elements are all equal to $-1$.

\begin{lemma}\label{lm-2}
The set
\[
\{ X^iY^j: 0 \le i \le p -2, 1 \le j \le p -1\}
\]
is a $\mathbb{Q}$-basis of $\mathbb{Q}^{(p-1) \times (p-1)}$. In particular, $X$ and $Y$ are generators of $\mathbb{Q}^{(p-1) \times (p-1)}$ as a $\mathbb{Q}$-algebra.
\end{lemma}
\begin{proof}
Let $\beta$ be a $p$-th primitive root of unity. According to \eqref{eq:poly-matrix}, there exists an isomorphism of algebras:
\[
\varphi: \mathbb{Q}(\beta)[x,\sigma]/(x^{p-1} - 1) \xrightarrow{\sim} \mathbb{Q}^{(p-1)\times (p-1)}.
\]
It is straightforward to verify that $\varphi(x) = X$ and $\varphi(\beta) = Y$ and our claim follows easily.
\end{proof}

\begin{corollary}
$\sum_{j=1}^{p-1}Y^j=-I$.
\end{corollary}
\begin{proof}
Since $\beta^{p-1} +\cdots+\beta+1=0$, we have $Y^{p-1} +\cdots+Y+I = \varphi(\beta^{p-1} +\cdots+\beta+1)=0$.
\end{proof}

\begin{proposition}
The circulant matrices are in a one-to-one correspondence with elements in $\mathbb{Q}(\beta)[x,\sigma]/(x^{p-1} - 1)$ with coefficients in $\mathbb{Q}$.
\end{proposition}

\begin{proof}
Assume $f=\sum_{j=0}^{p-2}\mu_{j+1} x^j$ is a skew polynomial with coefficients in  $\mathbb{Q}$. Then
\[
\varphi(f)=\sum_{j=0}^{p-2}\mu_{j+1} \varphi(x^j) = \sum_{j=0}^{p-2} \mu_{j+1} X^j.
\]
Since $X$ is circulant, $X^j$ is also circulant for each $j=0,\dots, p-2$. This implies that $\varphi(f)$ is a circulant matrix.

Conversely, we observe that each circulant matrix $C$ can be uniquely written as $C= \sum_{j=0}^{p-2} c_{j+1} X^j$ with $c_i\in \mathbb{Q}, j=1,\dots, p-1$, since $\{X^j\}_{j=0}^{p-2}$ is a basis of the space of all $(p-1)\times (p-1)$ circulant matrices over $\mathbb{Q}$. This implies that $\varphi^{-1}(C) = \sum_{j=0}^{p-2} c_{j+1} x^j$ is a skew polynomial whose coefficients are in $\mathbb{Q}$.
\end{proof}

For each $0 \le i \le p -2$, we define a $(p-1)$-dimensional subspace of matrices:
\[
\mathcal{L}_i \coloneqq \spn \left\lbrace X^i Y^j:0 \le j \le p-2 \right\rbrace.
\]
\begin{defn}
Given a matrix $A\in \mathbb{Q}^{(p-1)\times (p-1)}$, the \emph{skew-sparsity} of $A$ is the smallest positive integer $s$ such that $A \in \bigoplus_{i\in I} \mathcal{L}_i$ for some $I \subseteq \{0,\dots, p-2\}$ with $| I | = s$. Equivalently, the skew-sparsity of $A$ can also be defined as the sparsity of $\varphi^{-1}(A)$, where $\varphi$ is the map defined in \eqref{eq:poly-matrix}.
\end{defn}
As a direct consequence of Proposition~\ref{prop:matrix multiplication}, we have the following:
\begin{theorem}\label{the-3}
Let $p$ be a prime number and let $A,B$ be $(p-1) \times (p-1)$ matrices over $\mathbb{Q}$. If $A \in \bigoplus_{i\in I} \mathcal{L}_i$ and $B \in \bigoplus_{k\in K} \mathcal{L}_k$ for some subsets $I,K\subseteq \{0,\dots, p - 2\}$, then the product $AB$ can be computed by an algorithm with complexity $\O^\sim(T^{\omega-2} p^2)$ where $T$ is the cardinality of $I + K \pmod{p-1}$. In particular, if there exists some $\epsilon > 0$ such that
\[
|I||K| \le O(p^{1-\epsilon/(\omega-2)}),
\]
then the product $AB$ can be computed by an algorithm with complexity $\O(p^{\omega - \epsilon})$.
\end{theorem}

\subsection{Structure of $\mathcal{L}_i$}
This subsection is devoted to the discussion of the structure of $\mathcal{L}_i$. By definition, we have
\[
\mathcal{L}_i = X^i \mathcal{L}_0,\quad i=0,\dots, p-2.
\]
Observing that for any $A = \begin{bmatrix}
a_1,a_2, \dots, a_{p-1}
\end{bmatrix}^\tp \in \mathbb{Q}^{(p-1)\times (p-1)}$, we have
\[
X A  = \begin{bmatrix}
a_2, \dots, a_{p-1}, a_1
\end{bmatrix}^\tp.
\]
This implies that the effect of left multiplication by $X$ is the cyclic shift up of rows. Thus it suffices to consider the structure of $\mathcal{L}_0$, which consists of $(p-1)\times (p-1)$ matrices over $\mathbb{Q}$ corresponding to elements in $\mathbb{Q}(\beta) \subseteq \mathbb{Q}(\beta)[x,\sigma]$, via the isomorphism $\varphi$ defined in \eqref{eq:poly-matrix}.

According to \eqref{eq-2}, the matrix $Y$ satisfies the relation
\[
\begin{bmatrix}
v_1 \\
\vdots \\
v_{p-1}\\
\end{bmatrix}
\beta=Y
\begin{bmatrix}
v_1 \\
\vdots \\
v_{p-1}\\
\end{bmatrix}.
\]
By induction, we may further derive
\begin{equation}\label{eq:relation for Y}
\begin{bmatrix}
v_1 \\
\vdots \\
v_{p-1}\\
\end{bmatrix}
\beta^i=Y^i
\begin{bmatrix}
v_1 \\
\vdots \\
v_{p-1}\\
\end{bmatrix},\quad i =0,\dots, p-2.
\end{equation}

With \eqref{eq:relation for Y}, we are able to characterize $\mathcal{L}_0$. To that end, we notice that an element $a\in \mathbb{Q}(\beta)$ can be written as
\[
a=c_1\beta+c_2\beta^2+\cdots+c_{p-1}\beta^{p-1}
\]
for some $c_1,\dots, c_{p-1}\in \mathbb{Q}$. We introduce two permutations $s,q\in \mathfrak{S}_{p-1}$ of $\{1,\dots,p-1\}$ defined by
\begin{align*}
r^{s(i)-1}+i & \equiv 0 \pmod{p},\quad i=1,\dots,p-1, \\
r^{q(i)-1}-i & \equiv 0 \pmod{p},\quad i=1,\dots,p-1.
\end{align*}

\begin{lemma}\label{lm-7}
For any $1 \le i \le p-1$, the $s(i)$-th row of $Y^j$ is
\begin{equation*}
(Y^j)_{s(i)}=\left\{
\begin{aligned}
    & E_{q(j-i)}, \quad if\quad  j-i>0\\
    & E_{q(p+j-i)}, \quad if\quad  j-i<0\\
    &J,\quad\quad\quad if\quad  j-i= 0
\end{aligned}
\right.
\end{equation*}
where $J=(-1,\dots,-1)$ and $E_k$ is the row vector whose entries are all zero except the $k$-th, which is $1$.
\end{lemma}
\begin{proof}
By definition of $\varphi$, elements in the $s(i)$-th row of $Y^j$ are $\mathbb{Q}$-coefficients in the expansion of $v_{s(i)} \beta^j \in \mathbb{Q}(\beta)$ with respect to the basis $\{v_1,\dots,v_{p-1} \}$. Since $v_{s(i)} \beta^j=\beta^{r^{s(i)-1}+j}$ and $r^{s(i)-1} \equiv -i \pmod{p}$, we have $v_{s(i)} \beta^j= \beta^{j-i}$. We conclude the argument by the following three cases:
\begin{itemize}
\item if $j-i> 0$, then $\beta^{j-i} = v_{q(j-i)}$, thus $(Y^j)_{s(i)}= E_{q(j-i)}$;
\item if $j-i< 0$, then $\beta^{j-i}=\beta^{p+j-i}=v_{q(p+j-i)}$ and hence $(Y^j)_{s(i)}= E_{q(p+j-i)}$;
\item if  $j-i=0$, then $\beta^0=-v_1-\cdots-v_{p-1}$ and $(Y^j)_{s(i)}=J$.
\end{itemize}
\end{proof}

\begin{corollary}
For any $1 \le i \le p -1$ and $c_1,\dots, c_{p-1} \in \mathbb{Q}$, the $s(i)$-th row of $\varphi(\sum_{j=1}^{p-1} c_j \beta^j)$ is
\[
\left( \sum_{j=1}^{i-1} c_j E_{q(p+j - i)} \right) + c_i J + \left( \sum_{j=i+1}^{p-1} c_j E_{q(j-i)} \right).
\]
\end{corollary}
\begin{proof}
By the linearity of $\varphi$, we have $\varphi(\sum_{j=1}^{p-1} c_j \beta^j) = \sum_{j=1}^{p-1} c_j Y^{j}$ and the conclusion follows directly from Lemma \ref{lm-7}.
\end{proof}

Next we consider
\begin{equation}\label{eq-5}
A=
\begin{bmatrix}
E_{s(1)} \\
E_{s(2)} \\
\vdots \\
E_{s(n)}
\end{bmatrix},\quad
B=
\begin{bmatrix}
E_{q(1)} \\
E_{q(2)} \\
\vdots \\
E_{q(n)}
\end{bmatrix}^\tp.
\end{equation}
Since $r^{s(i)-1}+i \equiv 0 \pmod{p}$ and $r^{q(i)-1} \equiv i \pmod{p}$, we have
\[
r^{s(i)}+r^{q(i)}=0 \pmod{p},
\]
We notice that $r$ is a primitive root of $\mathbb{Z}_p$, thus $r^{\frac{p-1}{2}} \equiv -1 \pmod{p}$ and
\begin{equation}\label{eq-7}
s(i) \equiv q(i)+\frac{p-1}{2} \pmod{p-1}.
\end{equation}
According to \eqref{eq-7}, we obtain the next lemma describing a relation between matrices $A$ and $B$.
\begin{lemma}\label{lem:AB}
Let $A, B$ be as defined in \eqref{eq-5}. We have
 $AB= \begin{bmatrix}
0 & 0 & \cdots & 1\\
0 & 0 & \cdots & 0\\
\vdots & \vdots & \ddots & \vdots\\
1 & 0 & \cdots & 0\\
\end{bmatrix}$.
\end{lemma}
\begin{proof}
By definition, we have
$A B=(c_{ij})_{(p-1)\times (p-1)}$ where $c_{ij} = \delta_{s(i)q(j)},1\le i,j \le p-1$ and $\delta_{kl}$ is the Kronecker delta.
Equation \eqref{eq-7} implies
\[
c_{ij} = \delta_{q(i)+\frac{p-1}{2},q(j)},
\]
where $q(i)+\frac{p-1}{2}$ is regarded as an element in $\mathbb{Z}_{p-1}$. We observe that $q(i)+\frac{p-1}{2} \equiv q(j) \pmod{p-1}$ if and only if $r^{q(i)+\frac{p-1}{2}} \equiv r^{q(j)}\pmod{p}$. Since $r^{\frac{p-1}{2}} \equiv -1\pmod{p}$ and $r^{q(k)} \equiv k\pmod{p}$ for any $1\le k \le p-1$, we may further conclude that $q(i)+\frac{p-1}{2} \equiv q(j) \pmod{p-1}$ if and only if $-i \equiv j \pmod{p}$. Thus we obtain
\[
c_{ij} = \begin{cases}
1,\quad \text{if $p$ divides $(i+j)$}, \\
0,\quad \text{otherwise}.
\end{cases}
\]
\end{proof}

Given $c_1,\dots, c_{p-1} \in \mathbb{Q}$, we define
\begin{equation*}
P(c_1,\dots, c_{p-1}) \coloneqq \begin{bmatrix}
0 & c_{p-1}  & \cdots & c_3 & c_2  \\
c_1 & 0 &     \cdots  & c_4 & c_3    \\
\vdots & \vdots  &\ddots   & \vdots &\vdots \\
c_{p-3}& c_{p-4}  &     \cdots  & 0 & c_{p-1} \\
c_{p-2} & c_{p-3} &     \cdots  &c_1 &0  \\
\end{bmatrix},\quad
Q(c_1,\dots, c_{p-1})\coloneqq \begin{bmatrix}
c_1 & c_1 & \cdots & c_1\\
c_2 & c_2 & \cdots & c_2\\
\vdots & \vdots & \ddots & \vdots\\
c_{p-2} & c_{p-2} & \cdots & c_{p-2}\\
c_{p-1} & c_{p-1} & \cdots & c_{p-1}\\
\end{bmatrix}.
\end{equation*}
For instance, if $p = 7$ then we have
\begin{equation*}
P(c_1,\dots, c_6)= \begin{bmatrix}
0 & c_6 & c_5 & c_4 & c_3 & c_2\\
c_1 & 0 & c_6 & c_5 & c_4 &  c_3 \\
c_2 & c_1 & 0 & c_6 &  c_5 & c_4\\
c_3 & c_2 & c_1 &  0 & c_6 & c_5\\
c_4 & c_3 &  c_2 & c_1 & 0 & c_6\\
c_5 & c_4 & c_3 & c_2 &  c_1 & 0\\
\end{bmatrix},\quad
Q(c_1,\dots, c_6) = \begin{bmatrix}
c_1 & c_1 & c_1 & c_1 & c_1 & c_1\\
c_2 & c_2 & c_2 & c_2 & c_2 & c_2 \\
c_3 & c_3 & c_3 & c_3 & c_3 & c_3 \\
c_4 & c_4 & c_4 & c_4 & c_4 & c_4 \\
c_5 & c_5 & c_5 & c_5 & c_5 & c_5 \\
c_6 & c_6  & c_6 & c_6 & c_6  & c_6
\end{bmatrix}.
\end{equation*}

\begin{theorem}\label{thm:L_i}
The subspace $\mathcal{L}_0 \subseteq \mathbb{Q}^{(p-1) \times (p-1)}$ consists of matrices of the form
\[
A^{-1} \left( P(c_1,\dots,c_{p-1})-Q(c_1,\dots,c_{p-1}) \right) A,
\]
where $c_1,\dots, c_{p-1} \in \mathbb{Q}$. Moreover, for each $0 \le i \le p -2$, the subspace $\mathcal{L}_i$ is obtained by shifting up rows of elements in $\mathcal{L}_0$ by $i$. Here for a given $C = (c_{kl})_{k,l=1}^{p-1} \in \mathbb{Q}^{(p-1) \times (p-1)}$, the matrix obtained by shifting up rows of $C$ by $i$ is simply $(c_{k-i,j})_{k,l=1}^{p-1}$ and $(k-i)$ should be understood as $(k-i) \pmod{p-1}$.
\end{theorem}
\begin{proof}
Let $A,B$ be matrices defined in \eqref{eq-5}. Then for any $c_1,\dots, c_{p-1}\in \mathbb{Q}$, it is straightforward to verify the equation
\[
A \varphi \left( \sum_{j=1}^{p-1} c_j \beta^j\right) B = \left( P(c_1,\dots,c_{p-1})  - Q(c_1,\dots,c_{p-1}) \right) \begin{bmatrix}
0 & 0 & \cdots & 1\\
0 & 0 & \cdots & 0\\
\vdots & \vdots & \ddots & \vdots\\
1 & 0 & \cdots & 0\\
\end{bmatrix}.
\]
According to Lemma~\ref{lem:AB}, we obtain the desired characterization of $\mathcal{L}_0$.
\end{proof}

We conclude this section by a remark on the multiplication of matrices in $\mathcal{L}_0$. By Theorem~\ref{thm:L_i}, an element in $\mathcal{L}_0 $ can be written as $A (P(c_1,\dots,c_{p-1})  - Q(c_1,\dots,c_{p-1}) ) A^{-1}$ for some $c_1,\dots, c_{p-1} \in \mathbb{Q}$. On the one hand, we notice that in particular, $P(c_1,\dots,c_{p-1})$ is a Toeplitz matrix and $Q(c_1,\dots,c_{p-1})$ is a rank one matrix. We recall that the product of any matrix with a rank one matrix can be evaluated by $\O(p^2)$ operations, while the exsiting optimal algorithm for the product of a matrix and a Toeplitz matrix has complexity $\O^\sim(p^2)$ \cite{GV13}. Therefore, the product of two matrices in $\mathcal{L}_0$ can be evaluated by $\O^\sim(p^2)$ operations. Unfortunately, this method can not efficiently multiply elements in $\bigoplus_{i\in I} \mathcal{L}_i$. On the other hand, according to Theorem~\ref{the-3}, not only does Algorithm~\ref{alg-3} has complexity $\O^\sim(p^2)$ as well when applied to elements in $\mathcal{L}_0$, but it is also efficient on $\bigoplus_{i\in I} \mathcal{L}_i$ for any $I\subseteq \{0,\dots, p-2\}$.
%%%%%%%%%%%%%%%%%%%%%%%%%%%%%%%%%%%%%%%%%%%%%%%%%%%%
\section{randomized matrix multiplication via skew polynomials}
According to the proof of Proposition~\ref{prop:matrix multiplication}, we notice that the complexity of Algorithm~\ref{alg-3} is dominated by the cost of multiplying two skew polynomials. Here the support of the product is estimated by the sumset, which is guaranteed by Lemma~\ref{lem:sparsity} \eqref{lem:sparsity:item1}. However, this only supplies a rough estimate in the occurrence of cancellations of coefficients. For instance, if we take $f(x) = 1 - x$ and $g(x) = \sum_{j=0}^{k}x^j$ where $0 \le k \le p-3$, then obviously we have $\#(f\ast_{\sigma}g) = 2$ while $\#\mathbb{S}(f,g) = k+2$. This example also indicates that the sparsity of the product can be much smaller than the cardinality of the sumset.

The goal of this section is to present a randomized algorithm for matrix multiplication. In comparison to Algorithm~\ref{alg-3}, this randomized algorithm gets an acceleration by multiplying two skew-polynomials with the sparsity of their product known a priori. Essentially the randomized algorithm consists of two parts. One is the estimation of the sparsity of the product and the other is the interpolation of the product. It is worth to remark that the randomness in this algorithm is caused by the probabilistic method we use to estimate the sparsity of the product.
%%%%%%%%%%%%%%%%%%%%%%%%%%%%%%%%%%%%%%%%%%%%%%%%%%%%
\subsection{testing $AB \stackrel{?}{=} C$}\label{subsec:equalit test} In our randomized algorithm, we need to verify whether our result is correct. To do that, it suffices to test whether $AB = C$ for arbitrary matrices $A,B$ and $C$. We recall the celebrated Freivalds' algorithm \cite{Freivalds77} in Algorithm~\ref{alg-5}.
\begin{algorithm}[!ht]
\caption{Freivalds' algorithm}
\label{alg-5}
\begin{algorithmic}[1]
\renewcommand{\algorithmicrequire}{\textbf{Input}:}\Require matrices $M,A,B\in \mathbb{Q}^{n \times n}$ and $\mu\in(0,1)$.

\renewcommand{\algorithmicensure}{\textbf{Output}:}\Ensure ``equal" or ``not equal".
\State \label{alg-5:step1} compute $k:=\lceil \log_2\frac{1}{\mu}\rceil$.
\For{$i=1,\dots,k$} \label{alg-5:step2}

\State\label{alg-5:step2.1} randomly and uniformly choose $y = (y_1,\dots,y_n)\in \{0,1\}^n$.

\State\label{alg-5:step2.2} compute $\xi \coloneqq My$ and $\eta \coloneqq A By$.

\State if $\xi \neq \eta$, then return ``not equal".
\EndFor
\State return ``equal".
\end{algorithmic}
%\label{alg-5}
\end{algorithm}

By a straightforward calculation one can obtain the following error analysis.
\begin{lemma}\cite[Theorem~1.4]{Cryan2006}
\label{the-2}
Let $M,A,B\in \mathbb{Q}^{n\times n}$ be matrices. Assume $M\neq A B$. If $y \in \{0,1\}^n$ is uniformly and randomly chosen, then
\[
\mathbf{Pr}(My \neq A B y )\geq \frac{1}{2}.
\]
\end{lemma}

\begin{theorem}\label{the-1}
If $M=A B$, Algorithm~\ref{alg-5} returns ``equal". If $M\neq A B$, it returns ``not equal" with probability at least $(1-\mu)$, and it returns ``equal" with probability at most $\mu$. Assuming that we may obtain a random bit with bit-cost $\O(1)$, the cost of Algorithm~\ref{alg-5} is $\O(n^2\log\frac{1}{\mu})$ over $\mathbb{Q}$.
\end{theorem}
\begin{proof}
If $M=A B$, then $\xi = \eta$ and thus Algorithm~\ref{alg-5} returns ``equal". If $M\neq A B$, we observe that Algorithm~\ref{alg-5} returns ``not equal" if and only if in at least one iteration of Step~\ref{alg-5:step2}, $\xi \ne \eta$. In each iteration of Step~\ref{alg-5:step2}, Lemma~\ref{the-2} indicates that the probability of $\xi = \eta$ is at most $1/2$.  Since all iterations are independent, the probability for Algorithm~\ref{alg-5} to return ``not equal" is thus at least
\[
1-(\frac{1}{2})^k= 1-(\frac12)^{\lceil \log_2\frac{1}{\mu}\rceil}\geq 1-(\frac12)^{\log_2\frac{1}{\mu}}=1-\mu.
\]

Next we analyse the complexity. By assumption, the cost of obtaining $kn$ random bits is $\O(kn)$, while computing $\xi$ and $\eta$ in each iteration of Step~\ref{alg-5:step2} costs $\O(n^2)$ operations over $\mathbb{Q}$. Since $k  = \lceil \log_2 \frac{1}{\mu} \rceil$, we may conclude that the total complexity is $\O(n^2\log\frac{1}{\mu})$ over $\mathbb{Q}$.
\end{proof}

\subsection{interpolation}\label{subsec:interpolation}
In this subsection, we discuss the interpolation of a skew polynomial $f\in\mathbb{Q}(\beta)[x;\sigma]/(x^{p-1}-1)$ when an upper bound of its sparsity is given. Let $T\ge t\coloneqq \# (f) $ be a positive integer. We want to recover $f$ from its evaluations $f(v^{\ell}_1), \ell =0,1\dots,2T-1$.

Since the sparsity of $f$ is at most $T$, we can write $f(x) = \sum_{i=1}^T c_i x^{e_i}$. It is clear that there are exactly $t$ nonzero coefficients in $f$. We consider an auxiliary polynomial $\Lambda_T(z)$ defined as follows:
 \begin{equation}\label{eq-6}
    \Lambda_T(z) \coloneqq \prod_{i=1}^{T} (z-v_{e_j+1}) = z^T + \lambda_{T-1}z^{T-1}+\dots+\lambda_1 z+\lambda_0.
 \end{equation}
For any $\ell = 0,\dots, 2T-1$, we have
\[
  a_\ell \coloneqq f(v^\ell_1) = \sum_{i=1}^T c_i\sigma^{e_i}(v^{\ell}_1)  = \sum_{i=1}^T c_i v^\ell_{e_i+1}.
\]
We observe that for each $j=0, \dots, T - 1$,
we have
\small
\begin{align*}
\sum_{i=1}^{T} c_i v_{e_i + 1}^j \Lambda_T (v_{e_i + 1}) &=\left( \sum_{k=0}^{T - 1} \lambda_k(c_1 v_{e_1+1}^{k+j} +c_2 v_{e_2+1}^{k+j}+\dots+c_T v_{e_T+1}^{k+j})\right) +(c_1 v_{e_1+1}^{T+j}+c_2 v_{e_2+1}^{T+j} +\dots+c_T v_{e_T+1}^{T+j}) \\
&= \left( \sum_{k=0}^{T-1} a_{j + k} \lambda_k  \right) +  a_{j + T}.
\end{align*}
\normalsize

According to \eqref{eq-6}, we have $\Lambda_T (v_{e_i+1})=0$ for each $1\le i \le T$. This implies that
\[
\left( \sum_{k=0}^{T - 1} a_{j + k} \lambda_k  \right) +  a_{j + T} = 0,\quad 0\leq j\leq T - 1.
\]

We now have the Toeplitz system $A_T \vec{\lambda} =\vec{b}$  where
\begin{equation}\label{eq-toep}
A_T = \begin{bmatrix}
a_{T-1}&a_T&\cdots&a_{2T-2}\\
a_{T-2}&a_{T-1}&\cdots&a_{2T-3}\\
\vdots&\vdots&\ddots&\vdots\\
a_{0}&a_{1}&\cdots&a_{T-1}
\end{bmatrix},\quad \vec{\lambda} = \begin{bmatrix}
\lambda_0 \\
\lambda_1 \\
\vdots \\
\lambda_{T -1}\\
\end{bmatrix},\quad \vec{b} = -\begin{bmatrix}
a_{2T-1} \\
a_{2T-2} \\
\vdots \\
a_T\\
\end{bmatrix}.
\end{equation}

The matrix $A_T$ has rank $t$ which can be seen from the factorization:
\begin{equation}\label{eq-vand2}
A_T = \begin{bmatrix}
v_{e_1+1}^{T-1}& v_{e_2+1}^{T-1}&\cdots& v_{e_T+1}^{T-1}\\
v_{e_1+1}^{T-2}& v_{e_2+1}^{T-2}&\cdots& v_{e_T+1}^{T-2}\\
\vdots&\vdots&\ddots&\vdots\\
1&1&\cdots&1
\end{bmatrix}
\begin{bmatrix}
c_1&0&\cdots&0 \\
0&c_2&\cdots&0 \\
\vdots&\vdots&\ddots&\vdots \\
0&0&\cdots&c_T
\end{bmatrix}
\begin{bmatrix}
1& v_{e_1+1}&\cdots&v_{e_1+1}^{T-1} \\
1&v_{e_2+1}&\cdots&v_{e_2+1}^{T-1}\\
\vdots&\vdots&\ddots&\vdots \\
1&v_{e_T+1}&\cdots&v_{e_T+1}^{T-1}
\end{bmatrix}.
\end{equation}
Since the $v_i$'s are pairwise distinct, the two Vandermonde matrices in \eqref{eq-vand2} are nonsingular. By assumption, there are exactly $t$ nonzero $c_i$'s, thus the diagonal matrix in \eqref{eq-vand2} is has rank $t$. In particular, $A_t$ is nonsingular.
%The roots of the polynomial $\Lambda(z)$ are exactly $v_{e_j+1},j=1,\dots, T$.

By choosing the first $t$ evaluations of $f$, we obtain the following transposed Vandermonde system $V\vec{c}=\vec{a}$ for the coefficients of $f$, where
\begin{equation}\label{eq-vand}
V = \begin{bmatrix}
1&1&\cdots&1\\
v_{e_1+1}&v_{e_2+1}&\cdots&v_{e_t+1}\\
\vdots&\vdots&\ddots&\vdots\\
v_{e_1+1}^{t-1}&v_{e_2+1}^{t-1}&\cdots&v_{e_t+1}^{t-1}
\end{bmatrix},\quad \vec{c} = \begin{bmatrix}
c_1 \\
c_2 \\
\vdots \\
c_t
\end{bmatrix},\quad \vec{a} = \begin{bmatrix}
a_0 \\
a_1 \\
\vdots \\
a_{t-1}
\end{bmatrix}.
\end{equation}

Now we are ready to present Algorithm~\ref{alg-4} for the interpolation of a skew polynomial with a given upper bound  on its sparsity. Main ingredients of Algorithm~\ref{alg-4} are:
\begin{itemize}
\item compute the roots $v_{e_1+1},\dots,v_{e_t+1}$ of $\Lambda_t(z)$ from the $\{a_i\}_{i=0}^{2T - 1}$;
\item determine the exponents $e_1,\dots,e_t$ from $\{v_{e_j+1}\}_{j=1}^{t}$;
\item compute the coefficients $c_1,\dots,c_t$ by $\{a_k\}_{k=0}^{t-1}$ and $\{v_{e_j+1}\}_{j=1}^{t}$.
\end{itemize}

\begin{algorithm}[!ht]
\caption{Interpolation with an upper bound of sparsity}
\label{alg-4}
\begin{algorithmic}[1]
\renewcommand{\algorithmicrequire}{\textbf{Input}:}\Require
evaluations $\{a_\ell = f(v^\ell_{1})\}_{\ell=0}^{2T-1}$ and upper bound $T \ge \# (f)$.
\renewcommand{\algorithmicensure}{\textbf{Output}:}\Ensure coefficients of $f$.
\State \label{alg-4:step2} compute $t = \rank (A_T)$.
\State \label{alg-4:step3} form $A_t,\vec{b}_t$ by \eqref{eq-toep} and solve the Toeplitz system $A_t \vec{\lambda}_t=\vec{b}_t$ for $\vec{\lambda}_t$.
\State \label{alg-4:step4} construct $\Lambda_t(z)$ by \eqref{eq-6} and compute the roots $v_{e_1+1},\dots, v_{e_t + 1}$ of $\Lambda_t(z)$ by fast multiple evaluations at $v_1,v_2, \dots,v_{p-1}$.
\State \label{alg-4:step5} form $V$ and $\vec{a}$ by \eqref{eq-vand} and
solve the transposed Vandermonde system $V\vec{c}=\vec{a}$ for $\vec{c}$.
\end{algorithmic}
\end{algorithm}

\begin{lemma}\label{lm-6}
Algorithm \ref{alg-4} computes coefficients of $f$ by $\O^\thicksim(p^2)$ arithmetic operations over $\mathbb{Q}$.
\end{lemma}
\begin{proof}
The correctness follows from the discussion before Algorithm \ref{alg-4}, thus it is sufficient to analyse the complexity. According to \cite{kaltofen1988improved}, the cost of Steps~\ref{alg-4:step2}, \ref{alg-4:step3} and \ref{alg-4:step5} are respectively $\O^\thicksim(T)$ over $\mathbb{Q}(\beta)$. In Step~\ref{alg-4:step4}, it requires $\O^\thicksim(p)$ operations over $\mathbb{Q}(\beta)$ by \cite{GG99}. Since $T\leq p-1$, the total complexity of Algorithm \ref{alg-4} is  $\O^\thicksim(p^2)$ over $\mathbb{Q}$.
\end{proof}

\subsection{Monte Carlo matrix multiplication via skew polynomials} Our randomized algorithm for matrix multiplication is obtained by assembling ingredients discussed in Subsections~\ref{subsec:relation}, \ref{subsec:equalit test} and \ref{subsec:interpolation}. We record our algorithm in Algorithm~\ref{alg-6}.
\begin{algorithm}[!ht]
\caption{Monte Carlo matrix multiplication}
\label{alg-6}
\begin{algorithmic}[1]
\renewcommand{\algorithmicrequire}{\textbf{Input}:}\Require
$A,B\in \mathbb{Q}^{(p-1) \times (p-1)}$ and $\nu\in(0,1)$
\renewcommand{\algorithmicensure}{\textbf{Output}:}\Ensure $AB$.
\State \label{alg-6:step1}  find a $p$-th primitive root of unity $\beta$ and a generator $r$ of $\mathbb{Z}_{p}$.
\State \label{alg-6:step2} compute $v_i \coloneqq \beta^{r^{i-1}}, i=1,\dots,p-1$.
\State \label{alg-6:step3} compute $\varphi^{-1}(A),\varphi^{-1}(B)$ by \eqref{eq-1} w.r.t the normal basis $\{v_1,\dots,v_{p-1}\}$.
\State \label{alg-6:step4} set $T = 1$ and $\tau = 0$.
\While{$\tau = 0$} \label{alg-6:step5}
\State \label{alg-6:step5-0} compute $a_\ell \coloneqq \varphi^{-1}(A)\ast_{\sigma} \varphi^{-1}(B)(v^i_1), \ell =0,1,\dots,2T-1$.
\State \label{alg-6:step5-1}  interpolate $f$ with input $(\{a_\ell\}_{\ell=0}^{2T-1},T)$ by Algorithm \ref{alg-4}.
\State \label{alg-6:step5-2} compute $M \coloneqq \varphi(f)$ by Algorithm \ref{alg-2} w.r.t. the normal basis $\{v_1,\dots,v_{p-1}\}$.
\State \label{alg-6:step5-3} test whether $M=A B$ by calling Algorithm \ref{alg-5} with $\mu=\frac{\nu}{\lceil\log_2 (p-1) \rceil}$.
\State \label{alg-6:step5-4} if the test returns ``equal" then set $\tau = 1$; if the test returns ``not equal" then set $T \coloneqq 2T$.
\EndWhile
\end{algorithmic}
\end{algorithm}

\begin{theorem}
Given $A,B\in \mathbb{Q}^{(p-1) \times (p-1)}$ and $\nu \in (0,1)$, Algorithm~\ref{alg-6} computes $AB$ correctly with probability at least $(1 - \nu)$. The cost of Algorithm~\ref{alg-6} is $\O^\thicksim(t^{\omega-2}p^2+p^2\log\frac{1}{\nu})$ over $\mathbb{Q}$, where $t$ is the skew-sparsity of $AB$.
\end{theorem}
\begin{proof}
We denote by $t =\# \varphi^{-1}(A B)$ the skew-sparsity of $AB$. Assume $k$ is a positive integer such that $2^{k-1}< t\leq 2^k$. If Algorithm~\ref{alg-6} computes $AB$ correctly, then according to Lemma~\ref{lm-6}, the while loop in Step~\ref{alg-6:step5} of Algorithm~\ref{alg-6} must terminate at the $s$-th iteration, where $s \le k$. We denote by $M_i$ (resp. $\tau_i,T_i$) the matrix (resp. numbers) obtained in Step~\ref{alg-6:step5-2} (resp. Step~\ref{alg-6:step5-4}) at the $i$-th iteration, $i=1,\dots, s$. Then clearly we have
\begin{itemize}
\item $T_i = 2^i, i =1,\dots, s$;
\item $M_i \ne AB, i =1,\dots, s-1$ and $M_s = AB$;
\item $\tau_1 =\cdots = \tau_{s-1} = 0$ and $\tau_s = 1$.
\end{itemize}
By Theorem~\ref{the-1}, we obtain
\[
\mathbf{Pr}\left(
M_i \ne AB, \tau_i = 0, i=1,\dots, s-1
\right) \ge \left( 1-\frac{\nu}{\lceil\log_2 (p-1)\rceil} \right)^{s-1}\ge  1-\frac{(s-1)\nu}{\lceil\log_2 (p-1)\rceil}\geq 1- \nu.
\]
Therefore, the probability of Algorithm~\ref{alg-6} computing $AB$ correctly is at least $(1 - \nu)$.

It remains to analyse the complexity. First of all, if $r$ is a primitive root of $\mathbb{Z}_p$, then $r^1,r^2,\dots, r^{p-1}\pmod{p}$ must be distinct integers from $1$ to $(p-1)$. Since the complexity of checking the membership of a number in $\{2,\dots, p-1\}$ is $\O(p^2)$ over $\mathbb{Q}$, the total cost of Steps~\ref{alg-6:step1} and \ref{alg-6:step2} is $\O(p^2)$. By Proposition \ref{pro-1}, the complexity of Step~\ref{alg-6:step3}  is $\O^\thicksim(p^2)$ over $\mathbb{Q}$.

Next we count the complexity of each iteration in Step~\ref{alg-6:step5}. According to the analysis in Section \ref{sub-sec-1}, the complexity of  Step~\ref{alg-6:step5-0} is $\O^\sim(t^{\omega-2}p^2)$ over $\mathbb{Q}$. By Lemma \ref{lm-6} and Proposition \ref{pro-2}, the total cost of Steps~\ref{alg-6:step5-1} and \ref{alg-6:step5-2} is $\O^\thicksim(p^2)$ operations in $\mathbb{Q}$. Theorem \ref{the-1} impies that Step~\ref{alg-6:step5-3} requires $\O^\thicksim(n^2\log\frac{1}{\nu})$ operations in $\mathbb{Q}$.

Lastly, since there are only at most $\log_2 (p-1)$ iterations, the complexity of Algorithm~\ref{alg-6} is $\O^\thicksim(t^{\omega-2}p^2+p^2\log\frac{1}{\nu})$ over $\mathbb{Q}$.
 \end{proof}

\section{Conclusion}
In this paper, we propose a new method for $(p-1)\times (p-1)$ matrix multiplication over $\mathbb{Q}$ via skew polynomials, where $p$ is a prime number. Via the sparsity of skew polynomials, we are able to define skew-sparse matrices and we obtain an explicit characterization of such matrices. Based on the new method, we design a deterministic algorithm for matrix multiplication, which attains an acceleration if the product is skew-sparse. We also propose a randomized algorithm which can  further accelerate the matrix multiplication.

%\bibliographystyle{abbrv}
%\bibliography{ssmm}

\end{document}